\documentclass[12pt,thmsa,onecolumn,laterpaper]{article}
\usepackage{amsfonts}

\usepackage{amsmath}
\usepackage[all]{xy}
\usepackage{graphicx}

\newenvironment{proof}{\noindent{\bf{Proof.}}}{\hfill$\Box$\newline}

\newtheorem{corollary}{Corollary}[section]

\newtheorem{example}{Example}[section]

\newtheorem{proposition}{Proposition}[section]
\newtheorem{theorem}{Theorem}[section]
\begin{document}

\title{Codes Satisfying the Chain Condition with a Poset Weights}
\author{Luciano Panek\thanks{%
Centro de Engenharias e Ci\^{e}ncias Exatas, UNIOESTE, Av.
Tarqu\'{i}nio Joslin do Santos, 1300, CEP 85870-650, Foz do
Igua\c{c}u, PR, Brazil. Email:
lucpanek@gmail.com} \and Marcelo Firer%
\thanks{%
IMECC - UNICAMP, Universidade Estadual de Campinas, Cx. Postal 6065,
13081-970 - Campinas - SP, Brazil. E-mail: mfirer@ime.unicamp.br}}
\date{}
\maketitle

\begin{abstract}
In this paper we extend the concept of generalized Wei weights for
poset-weight codes and show that all linear codes $C$ satisfy the chain
condition if support of $C$ is a subposet totally ordered.

\vspace{0.5cm}

\end{abstract}

\section{Introduction}

In 1987 Niederreiter (\cite{HN1}) generalized the classical
problem of coding theory: find a linear code over a finite field
$\Bbb{F}_{q}$ of given length and dimension having the maximum
possible minimum Hamming distance (\cite{MS}). Later, in 1995,
Brualdi, Graves e Lawrence (\cite{RB}) extended the Niederreiter's
generalization introducing the notion of poset-codes. They started
the study of codes with a poset-weight, as we briefly introduce in
the next paragraph.

Let $\left( P,\leq \right) $ be a partially ordered finite set, abbreviated
as \textit{poset}, and assume $P=\left\{ 1,2,\ldots ,n\right\} $. An \textit{%
ideal} $I$ of $P$ is a subset of $P$ with the property that $y\in I$ and $%
x\leq y$, implies that $x\in I$. Given $A\subset P$, we denote by $%
\left\langle A\right\rangle $ the smallest ideal of $P$ containing $A$.
Given $x=\left( x_{1},\ldots ,x_{n}\right) \in \Bbb{F}_{q}^{n}$, the \textit{%
support} of $x$ is the set
\[
\mathrm{s}\text{\textrm{upp}}\left( x\right) =\left\{ i\in P:x_{i}\neq
0\right\} \text{.}
\]
We define the \textit{poset-weight} $w_{P}$ of $x$ (also called $P$\textit{%
-weight}), as the cardinality of the smallest ideal containing $\mathrm{s}$%
\textrm{upp}$\left( x\right) $, that is,
\[
w_{P}\left( x\right) =\left| \left\langle \mathrm{s}\text{\textrm{upp}}%
\left( x\right) \right\rangle \right| \text{.}
\]
The $P$-weight $w_{P}$ induces a metric in the vector space $\Bbb{F}_{q}^{n}$
defined by $d_{P}\left( x,y\right) =w_{P}\left( x-y\right) $ (%
\cite[Lemma 1.1]{RB}). If $P$ is antichain, i.e., $x\leq y$ iff $x=y$, then
the $P$-weight is usual Hamming weight $w_{H}$. An important family of poset
weights (non-Hamming weights) which can be applied to concrete communication
systems are Rosenbloom-Tsfasman weights (see \cite{RT}, \cite{SK}).

Motivated by several applications in cryptography, Wei introduced
in 1991 the concept of generalized Hamming weights (\cite{W}). We
extend here the concept of generalized Wei weights to
poset-weights. Let $P=\left\{ 1,2,\ldots ,n\right\} $ be a
partially ordered set. If $D$ is a linear subspace of the linear
code $C$ we write $D\leq C$. When $D$ is a proper
subspace of $C$ we write $D<C$. The \textit{generalized Hamming }$P$-\textit{%
weight} $\Vert \cdot \Vert _{P}$ of a $r$-dimensional subspace $D\leq \Bbb{F}%
_{q}^{n}$ is defined as
\[
\left\| D\right\| _{P}=\left| \bigcup_{x\in D}\left\langle \mathrm{s}\text{%
\textrm{upp}}\left( x\right) \right\rangle \right| \text{.}
\]
The $r$\textit{-th minimum Hamming }$P$-\textit{weight} of a $k$-dimensional
code $C\leq \Bbb{F}_{q}^{n}$ is
\[
d_{r}\left( C\right) =\min \left\{ \left\| D\right\| _{P}:D\leq C,\dim
\left( D\right) =r\right\} \text{.}
\]
A $k$-dimensional code $C\leq \Bbb{F}_{q}^{n}$ with $P$-weights hierarchy $%
\left( d_{1}\left( C\right) ,\ldots ,d_{k}\left( C\right) \right) $ is
called an $\left[ n;k;d_{1}\left( C\right) ,\ldots ,d_{k}\left( C\right)
\right] _{q}$\textit{-code}.

Many new perfect codes have been found with such poset-metrics (see \cite{HK}%
, for example). Motivated by this fact we investigated in this work the
possibility of the existence of new codes satisfying the chain condition
with the generalized $P$-weights. In the terminology of Wei and Yang (\cite
{WK}), a $k$-dimensional code $C\leq \Bbb{F}_{q}^{n}$ satisfies the \textit{%
chain condition} if there exists a sequence of linear subspaces (\textit{%
maximal flag})
\[
D_{1}<D_{2}<\ldots <D_{k-1}<D_{k}=C\text{,}
\]
with $\left\| D_{r}\right\| _{P}=d_{r}\left( C\right) $ and $\dim \left(
D_{r}\right) =r$ for all $r\in \left\{ 1,2,\ldots ,k\right\} $. In the case
that $P$ is antichain ($w_{P}=w_{H}$) the Hamming codes, dual Hamming codes,
Reed-Muller codes for all orders, maximum-separable-distance codes and Golay
codes satisfy the chain condition (see \cite{WK}). Moreover, every perfect
code must be a code satisfying the chain condition (see \cite{MS}).

In this work we will show that any poset-code $C\leq \Bbb{F}_{q}^{n}$ with
support totally ordered satisfies the chain condition. Moreover, the
sequence of linear subspaces $D_{1}<D_{2}<\ldots <D_{k-1}<D_{k}=C$ that
achieve the minimum Hamming $P$-weights is unique. It follows that if $%
\left\| D_{r}\right\| _{P}=d_{r}\left( C\right) $ for all $r\in \left\{
1,2,\ldots ,k\right\} $, then $D_{1}<D_{2}<\ldots <D_{k-1}<D_{k}=C$.

\section{Codes Satisfying the Chain Condition}

Before we show that any poset-code with support totally ordered
satisfies the chain condition, we give an example in the case that
$P$ is a weak order and shows the monotonicity of the minimum
poset-weights.

We denote by $spanX$ the linear subspace of $\Bbb{F}_{q}^{n}$ spanned by the
set $X\subset \Bbb{F}_{q}^{n}$.

\begin{example}
Let $W=n_{1}\mathbf{1}\oplus \ldots \oplus n_{9}\mathbf{1}$ the \emph{weak
order} given by the ordinal sum of the antichains $n_{1}\mathbf{1,}\ldots
,n_{9}\mathbf{1}$ with $3$ elements. Explicitly, $W=n_{1}\mathbf{1}\oplus
\ldots \oplus n_{9}\mathbf{1}$ is the poset whose underlying set and order
relation are given by
\[
\left\{ 1,2,\ldots ,27\right\} =n_{1}\mathbf{1}\cup \ldots \cup n_{9}\mathbf{%
1}\text{,}
\]
\[
n_{1}\mathbf{1=}\left\{ 1,2,3\right\} ,n_{2}\mathbf{1=}\left\{ 4,5,6\right\}
,\ldots ,n_{9}\mathbf{1=}\left\{ 25,26,27\right\}
\]
and $x<y$ if and only if $x\in n_{i}\mathbf{1}$, $y\in n_{j}\mathbf{1}$ for
some $i,j$ with $i<j$.

If $M_{9\times 3}\left( \Bbb{F}_{2}\right) $ is the linear space of all $%
9\times 3$ matrices over the finite field $\Bbb{F}_{2}$, we defined the
poset-weight $w_{W}$ of
\[
x=\left(
\begin{array}{ccc}
a_{1} & a_{2} & a_{3} \\
a_{4} & a_{5} & a_{6} \\
a_{7} & a_{8} & a_{9} \\
\vdots  & \vdots  & \vdots  \\
a_{22} & a_{23} & a_{24} \\
a_{25} & a_{26} & a_{27}
\end{array}
\right) \in M_{9\times 3}\left( \Bbb{F}_{2}\right)
\]
as $w_{W}\left( x\right) =w_{W}\left( a_{1},a_{2},a_{3},\ldots
,a_{25},a_{26},a_{27}\right) $.

In the Hamming space, the $\left[ 27;3;3,6,9\right] _{2}$-code
\[
C=span\left\{ \left(
\begin{array}{ccc}
1 & 0 & 0 \\
1 & 0 & 0 \\
1 & 0 & 0 \\
0 & 0 & 0 \\
0 & 0 & 0 \\
0 & 0 & 0 \\
0 & 0 & 0 \\
0 & 0 & 0 \\
0 & 0 & 0
\end{array}
\right) ,\left(
\begin{array}{ccc}
0 & 0 & 0 \\
0 & 0 & 0 \\
0 & 0 & 0 \\
0 & 1 & 0 \\
0 & 1 & 0 \\
0 & 1 & 0 \\
0 & 0 & 1 \\
0 & 0 & 0 \\
0 & 0 & 0
\end{array}
\right) ,\left(
\begin{array}{ccc}
0 & 0 & 0 \\
0 & 0 & 0 \\
0 & 0 & 0 \\
0 & 0 & 0 \\
0 & 0 & 0 \\
0 & 1 & 0 \\
0 & 0 & 1 \\
0 & 0 & 1 \\
0 & 0 & 1
\end{array}
\right) \right\}
\]
does not satisfy the chain condition. Now over the weak-metric space $%
M_{9\times 3}\left( \Bbb{F}_{2}\right) $, $C$ is a $\left[
27;3;7,19,25\right] _{2}$-code that satisfies the chain condition. We
observed that
\[
\mathrm{s}\text{\textrm{upp}}\left( C\right) =\left\{
1,4,7,11,14,17,21,24,27\right\}
\]
is totally ordered in weak order $W$ \emph{(}see Figure 1\emph{)}.

\begin{figure}[tph]
\centering
\includegraphics[width=1.10 in]{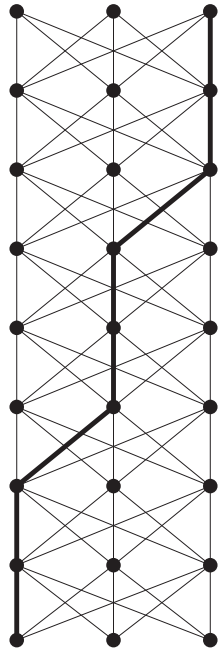}
\caption{Weak order $W=n_{1}\mathbf{1}\oplus \ldots \oplus n_{9}\mathbf{1}$.}
\end{figure}
\end{example}

As in \cite{W}, we have the monotonicity of the minimum
poset-weights.

\begin{proposition}
\label{prop.mon.}For any $\left[ n;k;d_{1}\left( C\right) ,\ldots
,d_{k}\left( C\right) \right] _{q}$-code $C\leq \Bbb{F}_{q}^{n}$ we have
that
\[
1\leq d_{1}\left( C\right) <d_{2}\left( C\right) <\ldots <d_{k}\left(
C\right) \leq n\text{.}
\]
\end{proposition}

\begin{proof}
We observed initially that $d_{r-1}\left( C\right) \leq
d_{r}\left( C\right) $. In fact, let $D_{r-1}$ and $D_{r}$
subcodes of $C$ with dimensions $r-1 $ and $r$ respectively such
that $\left\| D_{r-1}\right\| _{P}=d_{r-1}\left( C\right) $ and
$\left\| D_{r}\right\| _{P}=d_{r}\left( C\right) $. If $\left\|
D_{r-1}\right\| _{P}>\left\| D_{r}\right\| _{P}$, then for any
subcode $D_{r-1}^{\prime }<D_{r}$ of dimension $r-1$ we have that $%
\left\| D_{r-1}^{\prime }\right\| _{P}\leq \left\| D_{r}\right\|
_{P}<\left\| D_{r-1}\right\| _{P}=d_{r}\left( C\right) $. But this
contradiction the minimality of $d_{r}\left( C\right) $. We claim that the inequality $%
d_{r-1}\left( C\right) \leq d_{r}\left( C\right) $ is strict. Let
$D$ a subcode of $C$ with dimension $r$ such that $\left\|
D\right\| _{P}=d_{r}\left( C\right) $. If $i$ is a maximal element of $%
\mathrm{s}$\textrm{upp}$\left( D\right) $, then $D_{i}:=\left\{
v\in D:v_{i}=0\right\}$ is a subcode of $C$ with dimension $r-1$
such that
\[
d_{r-1}\left( C\right) \leq \left\| D_{i}\right\| _{P}\leq \left\| D\right\|
_{P}-1=d_{r}\left( C\right) -1\text{.}
\]
\end{proof}

Since $d_{r+1}(C)\geq d_{r}(C)+1$ and $d_{k}(C)\leq n$ we
immediately get the ge\-ne\-ra\-li\-zed Singleton bound:

\begin{corollary}
For an $\left[ n;k;d_{1}\left( C\right) ,\ldots ,d_{k}\left( C\right)
\right] _{q}$-code $C\leq \Bbb{F}_{q}^{n}$,
\[
r\leq d_{r}\left( C\right) \leq n-k+r\text{.}
\]
\end{corollary}

Now we will show that any poset-code $C$ with $\mathrm{s}$\textrm{upp}$%
\left( C\right) $ totally ordered satisfies the chain condition.

\begin{theorem}
\label{teo.chain}Let $C$ be a code in $\Bbb{F}_{q}^{n}$, endowed with a
poset-weight $w_{P}$. If $\mathrm{s}$\textrm{upp}$\left( C\right) $ is a
subposet of $P$ totally ordered then $C$ satisfies the chain condition.
\end{theorem}

\begin{proof}
Since $\mathrm{s}$\textrm{upp}$\left( C\right) $ is totally ordered, for
every $x,y\in C$ we have that $\left\langle \mathrm{s}\text{\textrm{upp}}%
(x)\right\rangle $ and $\left\langle \mathrm{s}\text{\textrm{upp}}%
(y)\right\rangle $ are contained one in the other. It follows that
\[
\left\| D\right\| _{P}=\left| \bigcup_{x\in D}\left\langle \mathrm{s}\text{%
\textrm{upp}}\left( x\right) \right\rangle \right| =\max \left\{
|\left\langle \mathrm{s}\text{\textrm{upp}}\left( x\right) \right\rangle
|:x\in D\right\} \text{,}
\]
so that for every $j\in \{1,2,\ldots ,k\}$ there is $v_{j}\in C$ such that $%
w_{RT}\left( v_{j}\right) =d_{j}\left( C\right) $. The set $%
\{v_{1},v_{2},\ldots ,v_{k}\}$ is linearly independent, because $w_{P}\left(
v_{1}\right) <\ldots <w_{P}\left( v_{k}\right) $ (see Proposition \ref
{prop.mon.}) and $\mathrm{s}$\textrm{upp}$\left( C\right) $ is totally
ordered. Consequently
\[
\dim \left( span\left\{ v_{1},v_{2},\ldots ,v_{j}\right\} \right) =j
\]
and
\[
span\left\{ v_{1}\right\} <span\left\{ v_{1},v_{2}\right\} <\ldots
<span\left\{ v_{1},v_{2},\ldots ,v_{k}\right\} =C\text{.}
\]
Since $\left\| span\left\{ v_{1},v_{2},\ldots ,v_{j}\right\} \right\|
_{P}=d_{j}\left( C\right) $ for every $j\in \left\{ 1,2,\ldots ,k\right\} $,
we find that $C$ satisfies the chain condition.
\end{proof}

\begin{theorem}
\label{teo.unic.}If $C\leq \Bbb{F}_{q}^{n}$, then there is a unique maximal
flag that achieve the generalized minimum Hamming $P$-wei\-ghts if $\mathrm{s%
}$\textrm{upp}$\left( C\right) $ is a subposet of $P$ totally ordered.
\end{theorem}

\begin{proof}
Let $k=\dim \left( C\right) $, $\left( d_{1}\left( C\right) ,d_{2}\left(
C\right) ,\ldots ,d_{k}\left( C\right) \right) $ the weights hierarchy of $C$
and $\left\{ e_{1},e_{2},\ldots ,e_{n}\right\} $ the canonical base of $\Bbb{%
F}_{q}^{n}$. Let $i_{1}<i_{2}<\ldots <i_{m}$ be the order in the $\mathrm{s}$%
\textrm{upp}$\left( C\right) $ and $D_{1}\leq C$ an $1$-dimensional subcode
of $C$ such that $\left\| D_{1}\right\| _{P}=d_{1}\left( C\right) $. We will
prove that $D_{1}$ is unique. In fact, let $D_{1}^{\prime }\leq C$ be an $1$%
-dimensional subcode of $C$ such that $\left\| D_{1}^{\prime }\right\|
_{P}=d_{1}\left( C\right) $ and $D_{1}^{\prime }\cap D_{1}=\left\{ \mathbf{0}%
\right\} $. Then there are $u\in D_{1}$ and $v\in D_{1}^{\prime }$ such that
\[
u=\alpha _{1}e_{i_{1}}+\ldots +\alpha _{r-1}e_{i_{r-1}}+e_{i_{r}}\text{,}
\]
\[
v=\beta _{1}e_{i_{1}}+\ldots +\beta _{r-1}e_{i_{r-1}}+e_{i_{r}}\text{,}
\]
with $\alpha _{j}\neq \beta _{j}$ for some $j\in \left\{ 1,2,\ldots
,r-1\right\} $ and $w_{P}\left( u\right) =w_{P}\left( v\right) =w_{P}\left(
e_{i_{r}}\right) =d_{1}\left( C\right) $. If
\[
l=\max \left\{ j\in \left\{ 1,2,\ldots ,r-1\right\} :\alpha _{j}\neq \beta
_{j}\right\} \text{,}
\]
it follow that $u+\left( q-1\right) v$ is a non zero vector of $C$ such that
\[
w_{P}\left( u+\left( q-1\right) v\right) =w_{P}\left( e_{i_{l}}\right)
<d_{1}\left( C\right) .
\]
But this contradicts the minimality condition of the $1$-th minimum Hamming $%
P$-weight of the code $C$. We conclude that $D_{1}$ is the unique subcode of
$C$ that achieve the $1$-th minimum Hamming $P$-weight of $C$.

The result follows now by induction on $\dim \left( D_{r}\right) =r$. Let $%
D_{1}<D_{2}<\ldots <D_{t-1}<C$, with $t-1<k$, be the sequence of linear
subspaces that achieve the $r$-th minimum Hamming $P$-weights of the code $C$
with $r\in \left\{ 1,2,\ldots ,t-1\right\} $, assures by Theorem \ref
{teo.chain}. Suppose that $D_{t}$ and $D_{t}^{\prime }$ are $t$-dimensional
subcodes of $C$ containing $D_{t-1}$ such that $D_{t}\neq D_{t}^{\prime }$
and $\left\| D_{t}\right\| _{P}=\left\| D_{t}^{\prime }\right\|
_{P}=d_{t}\left( C\right) $. Then there exist $w\in D_{t}$ and $z\in
D_{t}^{\prime }$ such that
\[
w=\gamma _{1}e_{i_{1}}+\ldots +\gamma _{s-1}e_{i_{s-1}}+e_{i_{s}}\text{,}
\]
\[
z=\eta _{1}e_{i_{1}}+\ldots +\eta _{s}e_{i_{s-1}}+e_{i_{s}}\text{,}
\]
with $\gamma _{j}\neq \eta _{j}$ for some $j\in \left\{ 1,2,\ldots
,s-1\right\} $ and $w_{P}\left( w\right) =w_{P}\left( z\right) =w_{P}\left(
e_{i_{s}}\right) =d_{t}\left( C\right) $. If
\[
p=\max \left\{ j\in \left\{ 1,2,\ldots ,s-1\right\} :\gamma _{j}\neq \eta
_{j}\right\} \text{,}
\]
then $x=w+\left( q-1\right) z$ is a non zero vector of $C$ such that $%
w_{P}\left( x\right) =w_{P}\left( e_{i_{p}}\right) <d_{t}\left( C\right) $
and $x\notin D_{t-1}$. Then, for every linearly independent subset $\left\{
y_{1},\ldots ,y_{t-1}\right\} $ $\subset D_{t-1}$, we find that $%
span\{y_{1},\ldots ,y_{t-1},x\}$ is a $t$-dimensional subspace of $C$ such
that $\left\| span\{y_{1},\ldots ,y_{t-1},x\}\right\| _{P}\leq d_{t}\left(
C\right) -1$, contradicting the minimality of $d_{t}\left( C\right) $.

By induction, the sequence of linear subspaces $D_{1}<D_{2}<\ldots <D_{k-1}<C
$ that achieve the $r$-th minimum Hamming $P$-weights of code $C$ is unique.
\end{proof}

\begin{corollary}
Let $C$ be a linear code in $\Bbb{F}_{q}^{n}$, endowed with a poset-weight $%
w_{P}$, such that $\mathrm{s}$\textrm{upp}$\left( C\right) $ is totally
ordered. If $D_{1},D_{2},\ldots ,D_{k-1},D_{k}=C$ is a sequence of subspaces
of $C$ such that $\left\| D_{r}\right\| _{P}=d_{r}\left( C\right) $ for all $%
r\in \left\{ 1,2,\ldots ,k\right\} $ and $\dim (D_{j})=j$, then $%
D_{1}<D_{2}<\ldots <D_{k-1}<C$.
\end{corollary}

If $P$ consists of finitely many disjoint union of $m$'s chains of lengths $%
n $ and $D$ is an $1$-dimensional subspace of the linear space $M_{n\times
m}\left( \Bbb{F}_{q}\right) $ of all $n\times m$ matrices over the finite
field $\Bbb{F}_{q}$, then $\left\| D\right\| _{P}$ become the
Rosenbloom-Tsfasman weight $w_{RT}$ (\cite{RT}), defined as follows: if $%
\left( a_{ij}\right) \in M_{n\times m}\left( \Bbb{F}_{q}\right) $, then
\[
w_{RT}\left( \left( a_{ij}\right) \right) =\sum_{j=1}^{m}\left| \left\langle
\mathrm{s}\text{\textrm{upp}}\left( a_{1j},a_{2j},\ldots ,a_{nj}\right)
\right\rangle \right|
\]
where
\[
\left\langle \mathrm{s}\text{\textrm{upp}}\left( a_{1j},a_{2j},\ldots
,a_{nj}\right) \right\rangle =\left\{ 1,2,\ldots ,i:i=\max \left\{
a_{ij}\neq 0\right\} \right\} \text{.}
\]

In the particular case of the Rosenbloom-Tsfasman spaces $M_{n\times
1}\left( \Bbb{F}_{q}\right) $, we have that:

\begin{corollary}
Every code $C\leq M_{n\times 1}\left( \Bbb{F}_{q}\right) $ satisfies the
chain condition.
\end{corollary}

To conclude we presented a lower bound for the number of codes satisfying
the chain condition.

\begin{proposition}
\label{prop.card.}Let $P=\left\{ 1,2,\ldots ,n\right\} $ be a poset and
suppose that $P=P_{1}\cup \ldots \cup P_{r}$ is a partition of $P$ into $r$
chains. If $\nu _{i}=\left| P_{i}\right| $ for all $i\in \left\{ 1,2,\ldots
,r\right\} $, then
\[
\sum_{i=1}^{r}\sum_{j=1}^{\nu _{i}}\prod_{k=1}^{j}\frac{\left( q^{\nu
_{i}}-q^{k-1}\right) }{\left( q^{j}-q^{k-1}\right) }
\]
is a lower bound for the number of codes satisfying the chain
condition in the space $\left( \Bbb{F}_{q}^{n},d_{P}\right) $. The
number $r$ is larger that or equal to the size of the largest
antichain in $P$.
\end{proposition}

\begin{proof}
For each subset $P^{\prime }\subseteq P$ we denoted by $\left[
P^{\prime
}\right] $ the subspace of $\mathbb{F}_{q}^{n}$ ge\-ne\-ra\-ted by the base $%
\{e_{i}\}_{i\in P^{\prime }}$, $e_{i}$ the canonical vector of $\Bbb{F}%
_{q}^{n}$. So, if $P={\bigcup }_{i=1}^{r}P_{j}$ is a partition
(disjoint union of sets) we have that $\mathbb{F}_{q}^{n}$ is a
direct sum
\[
\left[ P_{1}\right] \oplus \ldots \oplus \left[ P_{r}\right] \text{.}
\]
As $P_{i}$ is a chain for every $i\in \left\{ 1,2,\ldots
,r\right\} $, the Theorem \ref{teo.chain} assures that all code
$C\leq \left[ P_{i}\right] $ satisfies the chain condition.
Therefore the number of $j$-dimensional codes of $\left[
P_{i}\right] $ satisfying the chain condition equals
\[
\prod_{k=1}^{j}\frac{\left( q^{\nu _{i}}-q^{k-1}\right) }{\left(
q^{j}-q^{k-1}\right) }
\]
for each $i\in \left\{ 1,2,\ldots ,r\right\} $. This completes the proof of
the Proposition. The Dilworth's Theorem assures that $P$ can be partitioned
into $r$ chains if the largest antichain in the poset $P$ has size $r$.
\end{proof}


\begin{thebibliography}{1}
\bibitem[1]{RB}  R. Brualdi, J. S. Graves and M. Lawrence, \textit{Codes
with a poset metric}, Discrete Math. 147 (1995) 57-72.

\bibitem[2]{HK}  J. Y. Hyun and H. K. Kim, \textit{The poset
strutures admitting the extended binary Hamming code to be a
perfect code}, Discrete Math. 288 (2004) 37-47.

\bibitem[3]{HN1}  H. Niederreiter, \textit{A combinatorial problem for
vector spaces over finite fields}, Discrete Math. 96 (1991)
221-228.

\bibitem[4]{RT}  M. Yu Rosenbloom and M. A. Tsfasman, \textit{Codes for the }%
$m$\textit{-metric}, Probl. Inf. Transm. 33 (1997) 45-52.

\bibitem[5]{SK}  M. M. Skriganov, \textit{Coding theory and uniform
distributions}, St. Petesburg Math. J., vol. 13, n. 2, pp. 301-337, 2002.

\bibitem[6]{MS}  F. J. MacWilliams and N. J. A. Sloane, \textit{The Theory
of Error-Correcting Codes}, Amsterdam: North-Holland, 1977.

\bibitem[7]{W}  V. K. Wei, \textit{Generalized Hamming Weights for
Linear Code}, IEEE Trans. Inform. Theory, vol. \textbf{37}, n$^{\text{o}}\,$%
5, pp. 1412-1418, September 1991.

\bibitem[8]{WK}  V. K. Wei and K. Yang, \textit{On the
Generalized Hamming Weights for Product Code}, IEEE Trans. Inform. Theory,
vol. \textbf{39}, n$^{\text{o}}\,$5, pp. 1709-1713, September 1993.
\end{thebibliography}
\end{document}